\documentclass[10pt]{article}
\usepackage{amsfonts,color}
 \usepackage{amsmath,amssymb}
 \usepackage{epsfig}
\usepackage{lscape}
\usepackage{amssymb}
\usepackage{graphicx}
 \usepackage{footnote}
\usepackage[utf8]{inputenc}	
\usepackage{mathtools}
\usepackage{amsthm,amssymb,wasysym}
\usepackage{amsfonts}
\usepackage[english]{babel}

\newcommand{\Chi}{\raisebox{0.5ex}{\mbox{{\Large $\chi$}}}}
\newcommand{\Partial}{\raisebox{0ex}{\mbox{{\large $\partial$}}}}
\allowdisplaybreaks[1]

\makeindex

 \newtheorem{theorem}{Theorem}[section]
 \newtheorem{lemma}[theorem]{Lemma}
 \newtheorem{remark}[theorem]{Remark}

 \newtheorem{definition}[theorem]{Definition}
 \newcommand{\R}{\mathbb{R}}

 \newcommand{\tel}[1]{\frac{1}{#1}}

 \DeclareMathOperator{\diag}{diag}
 \DeclareMathOperator{\Sym}{Sym}
 \DeclareMathOperator{\PSym}{PSym}
 
 \DeclareMathOperator{\dev}{dev}

\def\barr{\begin{array}}
\def\id{1\!\!1}
\def\tr{\textrm{tr}}

\def\dd{\displaystyle}

\def\barr{\begin{array}}
\def\earr{\end{array}}
\def\bec#1{\begin{equation}\label{#1}}
\def\becn{\begin{equation*}}
\def\endec{\end{equation}}
\def\endecn{\end{equation*}}

\makeatletter
\let\@fnsymbol\@arabic
\makeatother

\begin{document}

\title{The exponentiated Hencky-logarithmic strain energy.\\ Part III: Coupling with idealized isotropic finite strain plasticity}
\author{
Patrizio Neff\thanks{Corresponding author: Patrizio Neff,  \ \ Head of Lehrstuhl f\"{u}r Nichtlineare Analysis und Modellierung, Fakult\"{a}t f\"{u}r
Mathematik, Universit\"{a}t Duisburg-Essen,  Thea-Leymann Str. 9, 45127 Essen, Germany, email: patrizio.neff@uni-due.de}\quad
and \quad
Ionel-Dumitrel Ghiba\thanks{Ionel-Dumitrel Ghiba, \ \ \ \ Lehrstuhl f\"{u}r Nichtlineare Analysis und Modellierung, Fakult\"{a}t f\"{u}r Mathematik,
Universit\"{a}t Duisburg-Essen, Thea-Leymann Str. 9, 45127 Essen, Germany;  Alexandru Ioan Cuza University of Ia\c si, Department of Mathematics,  Blvd.
Carol I, no. 11, 700506 Ia\c si,
Romania; and  Octav Mayer Institute of Mathematics of the
Romanian Academy, Ia\c si Branch,  700505 Ia\c si, email: dumitrel.ghiba@uni-due.de, dumitrel.ghiba@uaic.ro}}

\maketitle

\begin{center}
\textit{Dedicated to David J. Steigmann, a great scientist and good friend}
\end{center}
\bigskip

\begin{abstract}
   We investigate an immediate  application in finite strain multiplicative plasticity of the family of isotropic volumetric-isochoric decoupled strain energies
\begin{align*}
F\mapsto W_{_{\rm eH}}(F):=\widehat{W}_{_{\rm eH}}(U):=\dd\left\{\begin{array}{lll}
\dd\frac{\mu}{k}\,e^{k\,\|\dev_n\log {U}\|^2}+\dd\frac{\kappa}{{\text{}}{2\, {\widehat{k}}}}\,e^{\widehat{k}\,[\tr(\log U)]^2}&\text{if}& \det\, F>0,\vspace{2mm}\\
+\infty &\text{if} &\det F\leq 0,
\end{array}\right.\quad
\end{align*}
based on the Hencky-logarithmic (true, natural) strain tensor $\log U$. Here, $\mu>0$ is the infinitesimal shear modulus,
$\kappa=\frac{2\mu+3\lambda}{3}>0$ is the infinitesimal bulk modulus with $\lambda$ the first Lam\'{e} constant, $k,\widehat{k}$ are dimensionless fitting
parameters, $F=\nabla \varphi$ is the gradient of deformation,  $U=\sqrt{F^T F}$ is the right stretch tensor and $\dev_n\log {U} =\log {U}-\frac{1}{n}\,
\tr(\log {U})\cdot\id$
 is the deviatoric part  of the strain tensor $\log U$.

 Based on the multiplicative  decomposition $F=F_e\, F_p$, we couple these energies with some isotropic elasto-plastic flow rules $F_p\,\frac{\rm d}{{\rm d t}}[F_p^{-1}]\in-\Partial \Chi(\dev_3 \Sigma_{e})$  defined  in the plastic distortion $F_p$,
where $\Partial \Chi$ is the subdifferential of the indicator function $\Chi$ of the convex elastic domain $\mathcal{E}_{\rm e}(W_{\rm iso},{\Sigma_{e}},\frac{1}{3}{\boldsymbol{\sigma}}_{\!\mathbf{y}}^2)$ in the mixed-variant $\Sigma_{e}$-stress space and $\Sigma_{e}=F_e^T D_{F_e} W_{\rm iso}(F_e)$. While $W_{_{\rm eH}}$ may loose ellipticity, we show that   loss of ellipticity is effectively prevented by the coupling with plasticity, since the ellipticity domain of $W_{_{\rm eH}}$ on the one hand, and the elastic domain in $\Sigma_{e}$-stress space on the other hand, are closely related. Thus the new formulation remains elliptic in elastic unloading at any given plastic predeformation. In addition, in this domain, the true-stress-true-strain relation remains monotone, as observed in experiments.
\end{abstract}

\newpage

\tableofcontents

\section{Introduction}

\subsection{Preliminaries}

It is impossible to give an account of all works treating finite strain plasticity based in some way or another on the logarithmic strain space description. The logarithmic
description\footnote{According to Hanin and Reiner \cite[page 384]{hanin1956isotropic}: ``... there are problems in large plastic deformation. Here the
only adequate measure is the Hencky measure, because this is the only measure
in which the extensions form a group as can be seen from the relation
 $$\log \frac{\ell_3}{\ell_1} =\log\left( \frac{\ell_3}{\ell_2}\frac{\ell_2}{\ell_1}\right)=\log  \frac{\ell_3}{\ell_2}+\log\frac{\ell_2}{\ell_1}.$$
This property of forming a group is required in plasticity because in (ideal)
plasticity (as in viscosity) the amount of finite deformation reached at any
time is of no physical significance. As a matter of fact no definite meaning can
be attached to such deformation because while in elasticity there exists an
`unstrained state" to which the length $\ell_0$ is referred, no ``undeformed state" can
be defined. In plasticity, as in viscosity, the increase in length $d\ell$, which takes
place during the time-increment $dt$, can only be referred to the instantaneous
length $\ell$ so that the extension
$$
\varepsilon=\int_{\ell_{n}}^{\ell_{n+1}}\frac{dl}{l}=\log \frac{\ell_{n}}{\ell_{n+1}}
$$
which is Hencky's measure. At the same time, only in the Hencky measure can
the cubical dilatation be measured by the first invariant as can be seen from
$$
\varepsilon_v=\log \frac{V}{V_0}=\log\left(\frac{\ell_{i}}{\ell_{0i}}\frac{\ell_{j}}{\ell_{0j}}
\frac{\ell_{k}}{\ell_{0k}}\right)=\log\frac{\ell_{i}}{\ell_{0i}}+\log\frac{\ell_{j}}{\ell_{0j}}+\log \frac{\ell_{k}}{\ell_{0k}}
$$
so that the deviator is of physical significance; and plasticity relations must be
expressed in terms of deviators."}  is arguably the {simplest} approach to finite plasticity, suitable for the phenomenological description of isotropic
polycrystalline metals if the structure of geometrically linear theories {is} used with respect to the Lagrang{i}an logarithmic strain. In this paper we do not consider hypoelastic-plastic models \cite{ogden1970compressible,gurtin1983relationship,meyers2006choice,xiao1999existence} in which, contrary to hyperelastic models,
 the potential character of the elastic energy is ignored \cite{meyers1999consistency,muller2006thermodynamic}. Otherwise, they are simply the hyperelastic models
rewritten in a suitable incremental form. In case of the logarithmic rate, however, the hypo-elastic model integrates exactly to the well-known hyperelastic quadratic Hencky
model.

 In isotropic finite strain computational   hyperelasto-plasticity \cite{bertram2012elasticity,bertram1999alternative,glugegraphical,shutov2013analysis,shutov2008finite,reese2008finite,dettmer2004theoretical,Reese97a} the mostly used elastic energy is the quadratic Hencky logarithmic energy
\cite{Simo85,HutterSFB02,BatheEterovic,MieheApel,xiao2000consistent,zhu2014logarithmic,bruhns1999self,henann2009large} (see also \cite{heiduschke1995logarithmic,
sansour2001dual,sansour2003viscoplasticity,peric1992model,papadopoulos1998general,gabriel1995some,mosler2007variational})
\begin{align}
W_{_{\rm H}}(F_e):=\frac{\mu}{4\,}\|\dev_n \log C_e\|^2+\frac{\kappa}{8}\,[\tr( \log C_e)]^2={\mu}{}\|\dev_n \log U_e\|^2+\frac{\kappa}{2}\,[\tr( \log U_e)]^2,
\end{align}
where  $\mu>0$ is the infinitesimal shear modulus,
$\kappa>0$ is the infinitesimal bulk modulus, $C_e:=F_e^TF_e$ is the elastic right Cauchy-Green tensor,  $U_e$ the right
stretch tensor, i.e. the unique element of ${\rm PSym}(n)$ for which $U_e^2=C_e$  and
\begin{align}
F=F_e\cdot F_p
\end{align}
is the multiplicative decomposition of the deformation gradient \cite{kroner1955fundamentale,kroner1958kontinuumstheorie,kroner1959allgemeine,lee1969elastic,neff2009notes,Neff_Knees06,carstensen2002non,mielke2003energetic,mielke2002finite}. Here we have used  the Frobenius tensor norm
$\|{X}\|^2=\langle {X},{X}\rangle_{\R^{n\times n}}$, where $\langle {X},{Y}\rangle_{\R^{n\times n}}$ is the standard Euclidean scalar product on $\R^{n\times n}$. The identity tensor on $\R^{n\times n}$ will be denoted by $\id$, so that
$\tr{(X)}=\langle {X},{\id}\rangle$.

Among the works which use the Hencky strain in
elasto-plasticity we may also mention
\cite{simo1992algorithms,MieheApel,papadopoulos1998general,caminero2011modeling,peric1992model,peric1999new,dvorkin1994finite,geers2004finite}.
The expression $W_{_{\rm H}}$ is the energy considered by  J.C. Simo (see Eq. (3.4), page 147, from \cite{simo1993recent} and also \cite{armero1993priori})
because
\begin{align}
W_{_{\rm H}}(F_e):&=\mu\|\dev_n\log \sqrt{F_e^TF_e}\|^2+\frac{\kappa}{2}[\tr(\log \sqrt{F_e^TF_e})]^2=\frac{\mu}{4}\|\dev_n\log
{F_e^TF_e}\|^2+\frac{\kappa}{8}[\tr(\log F^T F)]^2\\
&=\frac{\mu}{4}\|\dev_n\log {F_e^TF_e}\|^2+\frac{\kappa}{2}[\log (\det F)]^2.\notag
\end{align}
As J.C. Simo already pointed out \cite[page 392]{Simo98a}, the Hencky energy  $W_{_{\rm H}}$ ``has the correct behaviour for extreme strains in the sense that" $W(F_e)\rightarrow\infty$ as $\det
F_e\rightarrow0$ and, likewise $W(F_e)\rightarrow\infty$ as $\det F_e\rightarrow\infty$, but $W_{_{\rm H}}$ ``is not a convex function of" $\det F_e$ ``and hence
$ W_{_{\rm H}}$ ``cannot be a polyconvex function of the deformation gradient [...]. Therefore, the stored energy function" $W_{_{\rm H}}$ ``cannot be accepted as a
correct model of elasticity for extreme strains. Despite this shortcoming, the model provides an excellent approximation for moderately large elastic
strains, vastly superior to the usual Saint-Venant-Kirchhoff model of finite elasticity\footnote{The isotropic  Saint-Venant-Kirchhoff elastic energy $W_{SV\!K}(F_e)
$ reads: $W_{SV\!K}(F_e)=\frac{\mu}{4}\|F_e^T\,F_e-\id\|^2+\frac{\lambda}{8}[\tr(F_e^T\,F_e-\id)]^2$ and does not satisfy the Baker-Ericksen-inequalities,
and is not separately convex. Therefore $W_{SV\!K}(F_e) $ is not rank-one convex \cite{Neff_Diss00,Raoult86,cism_book_schroeder_neff09}. Moreover, in the
neighbourhood of the identity it has the wrong nonlinear second order correction compared to all know experimental facts.  For this reason, $W_{SV\!K}(F_e)
$  is not a useful strain energy expression and should therefore be avoided in simulations.}. Furthermore, this
limitation has negligible practical implications in realistic models of classical plasticity, which are typically restricted to small elastic strains, and
is more than offset by the simplicity of the return mapping algorithm in stress space, which takes a format identical to that of the infinitesimal theory".
The last statement is the core argument why the Hencky energy is favoured in computational metal elasticity\footnote{We need to be a little more specific. For the additive model in the format $\|\log C-\log C_p\|^2$ the complete systems of equations of the plastic flow rule are {\bf identical} to the infinitesimal additive model, while for the truly multiplicative model the return mapping algorithm is {\bf similar} to the infinitesimal case. }.

  Several models of such a type have been considered in \cite{Neff_Wieners,krishnan2014polyconvex}. The decisive advantage of using the energy $W_{_{\rm H}}$ compared to other elastic
  energies stems from the fact that computational implementations of elasto-plasticity  \cite{gabriel1995some} based on the additive decomposition
  $\varepsilon=\varepsilon_e+\varepsilon_p$ in infinitesimal models \cite{neff2010parallel,ebobisse2010existence,neff2009numerical,neff2007numerical,chleboun16extension}, can
  be used with nearly no changes also in isotropic finite  strain problems \cite[page 392]{Simo98a}.

The computation of the elastic equilibrium at given plastic distortion $F_p$ suffers, however, under the well-known non-ellipticity\footnote{We know that
$W_{_{\rm H}}$ is LH-elliptic in a large neighbourhood of the identity if $\lambda,\mu>0$, $\lambda_i\in[0.21162...,1.39561...]$  (see \cite{Bruhns01,Bruhns02JE}),
therefore $F\mapsto W_{_{\rm H}}(F)$ is LH-elliptic for small elastic strains. Since the elasto-plastic model  should secure small elastic strains anyway it seems that the
non-ellipticity occurring for larger elastic  strains is not essential. However, in numerical FEM-implementation it is necessary to compute the so-called
elastic trial stress. The corresponding elastic trial deformation states $F_e$ may well be far outside the LH-ellipticity range. The model using the Hencky
energy $W_{_{\rm H}}$ cannot guarantee that  the computation of the elastic trial state is well-posed!} of $W_{_{\rm H}}$
\cite{Bruhns01,Balzani_Schroeder_Gross_Neff05,Neff_Diss00,Hutchinson82}.

 Recently, it has been discovered that the elastic Hencky energy does  have a fundamental differential geometric meaning, not shared by any other elastic energy. In fact
\begin{align*}
\widetilde{W}_{_{\rm H}}^{\rm iso}\left(\frac{U}{\det U^{1/n}}\right)={\rm dist}^2_{{\rm geod}}\left( \frac{F}{(\det F)^{1/n}}, {\rm SO}(n)\right)&={\rm dist}^2_{{\rm geod,{\rm SL}(n)}}\left( \frac{F}{(\det F)^{1/n}}, {\rm
SO}(n)\right)=\|\dev_n \log U\|^2 \,,\\
\widetilde{W}_{_{\rm H}}^{\rm vol}(\det U)={\rm dist}^2_{{\rm geod}}\left((\det F)^{1/n}\cdot \id, {\rm SO}(n)\right)&={\rm dist}^2_{{\rm geod,\mathbb{R}_+\cdot \id}}\left((\det F)^{1/n}\cdot \id,
\id\right)=|\log \det F|^2\,,
\end{align*}
where ${\rm dist}^2_{{\rm geod,\mathbb{R}_+\cdot \id}}$ and ${\rm dist}^2_{{\rm geod,{\rm SL}(n)}}$ are the canonical left invariant geodesic distances on
the Lie-group ${\rm SL}(n)$ and on the group $\mathbb{R}_+\cdot\id$, respectively (see
\cite{Neff_Osterbrink_Martin_hencky13,neff2013hencky,Neff_Nagatsukasa_logpolar13,RobertNeff}). For this investigation new mathematical tools had to be
discovered \cite{Neff_Nagatsukasa_logpolar13,LankeitNeffNakatsukasa} also having consequences for the classical polar  decomposition
\cite{jog2002explicit,jog2002foundations}. Here we adopt
the usual abbreviations of Lie-group theory and we let $\Sym(n)$ and $\rm PSym(n)$ denote the symmetric and positive definite symmetric tensors respectively. We  denote by $C=F^T F$ the right Cauchy-Green  tensor,  $U$ the right
stretch tensor, $B=F\, F^T$ the left Cauchy-Green (or Finger)  tensor,  and by  $V$ the  left stretch tensor.

In the remaining part of this paper, after a paragraph giving some information on the results obtained for the exponentiated Hencky energy,  we will  consider the coupling to finite plasticity based on a 9-dimensional flow rule \cite{steigmann2011mechanically,gupta2011aspects}.

\subsection{The exponentiated Hencky energy}

With a view to overcome the shortcomings of the quadratic Hencky energy, in a previous work \cite{NeffGhibaLankeit} we have modified the Hencky energy and we considered
\begin{align}\label{thdefHen}\hspace{-2mm}
 W_{_{\rm eH}}(F)=W_{_{\rm eH}}^{\text{\rm iso}}(\frac F{\det F^{\frac{1}{n}}})+W_{_{\rm eH}}^{\text{\rm vol}}(\det F^{\tel n}\cdot \id) = \left\{\begin{array}{lll}
\dd\frac{\mu}{k}\,e^{k\,\|{\rm dev}_n\log U\|^2}+\frac{\kappa}{2\widehat{k}}\,e^{\widehat{k}\,[(\log \det U)]^2}&\text{if}& \det\, F>0,\vspace{2mm}\\
+\infty &\text{if} &\det F\leq 0\,.
\end{array}\right.
\end{align}
We have called this the {\bf exponentiated Hencky energy}. For the  two-dimensional situation $n=2$ and for  $\mu>0, \kappa>0$,   we have established that the functions $W_{_{\rm eH}}:\R^{n\times n}\to \overline{\R}_+$ from the family of exponentiated Hencky type
energies are {\bf rank-one convex} \cite{NeffGhibaLankeit} for  $k\geq\dd\frac{1}{4}$ and $\widehat{k}\dd\geq \tel8$, while they are  {\bf polyconvex} \cite{NeffGhibaPoly} for  $k\geq\dd\frac{1}{3}$ and $\widehat{k}\dd\geq \tel8$.

Regarding the three-dimensional case we have proved \cite{NeffGhibaLankeit} that, for all $k>0$, the function
 $
F\mapsto e^{k\,\|\dev_3\log U\|^2},$ $ F \in{\rm GL}^+(3)
$
is not rank-one convex.
However, in the next section we will discuss an interesting relation between non-ellipticity of $ W_{_{\rm eH}}$ in three-dimensions and finite plasticity models.

We note that the  {\bf Kirchhoff stress tensor} $\tau_{_{\rm eH}}$ corresponding to the exponentiated energies is given \cite{Ogden83} by
 \begin{align}\label{exptau}
 \tau_{_{\rm eH}}=2\,{\mu}\,e^{k\,\|\dev_3\log\,V\|^2}\cdot \dev_3\log\,V+{\kappa}\,e^{\widehat{k}\,[\tr(\log V)]^2}\,\tr(\log V)\cdot \id,
 \end{align}
 while the \textbf{Cauchy stress tensor} is given by
 $$
 \sigma_{_{\rm eH}}=e^{-\tr(\log  V)}\cdot \tau_{_{\rm eH}}.
$$
 Both tensors $\sigma_{_{\rm eH}}$ and $\tau_{_{\rm eH}}$ differ from their classical Hencky-counterparts $\sigma_{_{\rm H}}$ and $\tau_{_{\rm H}}$ only by some nonlinear scalar factors.
 Moreover, by orthogonal projection onto the Lie-algebra $\mathfrak{sl}(3)$ and $\R\cdot \id$, respectively, we find
 \begin{align}
 \dev_3\sigma_{_{\rm eH}}&=e^{-\tr(\log  V)}\, \dev_3\sigma_{_{\rm H}},\qquad
 \tr(\sigma_{_{\rm eH}})=e^{-\tr(\log  V)}\, \tr(\sigma_{_{\rm H}}).
\end{align}
Therefore, the deviatoric part of the Cauchy stress $\dev_3\sigma_{_{\rm eH}}$ and the trace of the Cauchy stress
$\tr(\sigma_{_{\rm eH}})$ are in a simple relation with the corresponding quantities for the quadratic Hencky energy $W_{_{\rm H}}$. Hence, the change of a given FEM-implementation
of $W_{_{\rm H}}$ into $W_{_{\rm eH}}$ is nearly free of costs \cite{tanaka1979finite,masud1997finite,naghdabadi2012viscoelastic,heiduschke1996computational}.

We also need to introduce the convex  elastic domain in the {Kirchhoff-stress space}
\begin{align}\label{contintau2}
 \mathcal{E}_{\rm e}(\tau_{_{\rm e}},\frac{2}{3}\, {\boldsymbol{\sigma}}_{\!\mathbf{y}}^2):=\left\{ \tau_e\in {\rm Sym}(3) \big|\,\ \|\dev_3
 \tau_e\|^2\leq\frac{2}{3}\, {\boldsymbol{\sigma}}_{\!\mathbf{y}}^2\right\}\subset  {\rm Sym}(3).
 \end{align}
 Incidentally, the set $\mathcal{E}_{\rm e}(\tau_{_{\rm eH}},\frac{2}{3}\, {\boldsymbol{\sigma}}_{\!\mathbf{y}}^2)$ coincides with the set considered in the study   of the monotonicity properties of the map $\log U\mapsto \sigma_{_{\rm eH}}(\log U)$ which we have called the true-stress-true-strain (TSTS-M$^+$) monotonicity condition \cite{NeffGhibaLankeit,jog2013conditions}. The
    monotonicity of the Cauchy stress tensor as a function of $\log B$ or $\log V$ means
  \begin{align}\label{Jogines1}
 \langle\sigma(\log B_1)-\sigma(\log B_2),\log B_1-\log B_2\rangle> 0, \qquad \forall\, B_1, B_2\in \PSym^+(3), \ B_1\neq B_2,
 \end{align}
 which implies the \textbf{true-stress-true-strain-invertibility} (TSTS-I), i.e. the invertibility of  the map $\log B\mapsto \sigma(\log B)$.
  This  means that for our $W_{_{\rm eH}}$-formulation, the {\bf true-stress-true-strain}  relation is {\bf monotone inside the elastic domain} $\mathcal{E}_{\rm e}(\tau_{_{\rm eH}},\frac{2}{3}\, {\boldsymbol{\sigma}}_{\!\mathbf{y}}^2)$. This is a feature of $W_{_{\rm eH}}$ not shared  with any other known elastic energy. For more constitutive issues regarding the interesting properties of $W_{_{\rm eH}}$ we refer the reader to \cite{NeffGhibaLankeit}.

\section{Multiplicative isotropic elasto-plasticity directly in terms \\ of the non-symmetric plastic distortion $F_p$}\label{plastSect}
\setcounter{equation}{0}

In planar elasto-plasticity\footnote{In order to model plane strain with this model, the coefficient $\kappa$ has to be modified in order to be consistent with plane strain linear elasticity in the infinitesimal limit, see \eqref{3.2}.} our development  suggests to replace the energy $W_{_{\rm H}}$ by
\begin{align}
W_{_{\rm eH}}(F_e):&=\frac{\mu}{k}\, e^{\frac{k}{4}\|\dev_2 \log C_e\|^2}+\frac{\kappa}{2\widehat{k}}\,e^{\frac{\widehat{k}}{4}\,[\tr( \log C_e)]^2}\\
&=\frac{\mu}{k}\, e^{\frac{k}{4}\|\dev_2 \log F_p^{-T}F^T FF_p^{-1}\|^2}+\frac{\kappa}{2\widehat{k}}\,e^{\frac{\widehat{k}}{4}\,[\tr( \log
F^TF)]^2}=\frac{\mu}{k}\, e^{\frac{k}{4}\|\dev_2 \log F_p^{-T}F^T FF_p^{-1}\|^2}+\frac{\kappa}{2\widehat{k}}\,e^{\widehat{k}\,[ \log(\det F)]^2},\notag
\end{align}
 where  we have imposed the condition of plastic incompressibility $\det F_p=1$. Let us remark that for small deformations
\begin{align}\label{3.2}
W_{_{\rm eH}}(F_e)&=\frac{\mu}{k}\, e^{{k}\|\dev_2 \log U_e\|^2}+\frac{\kappa}{2\widehat{k}}\,e^{{\widehat{k}}\,[\tr( \log U_e)]^2}\notag={\mu}\, \|\dev_2 \log U_e\|^2+\frac{\kappa}{2}\,{\,[\tr( \log U_e)]^2}+\,\text{h.o.t.}\\
&= {\mu}\, \| \log U_e\|^2+\frac{\kappa-\mu}{2}\,{[\tr( \log U_e)]^2}+\,\text{h.o.t.}= {\mu}\, \| \varepsilon_e\|^2+\frac{\kappa-\mu}{2}\,{\,[\tr(\varepsilon_e)]^2}+\,\text{h.o.t.},
\end{align}
where  $\varepsilon_e$ is the symmetric plastic strain. A direct identification of the constitutive coefficients gives us that
\begin{align}
\mu=\mu_{_{\rm 3D}}, \qquad  \kappa-\mu=\lambda_{_{\rm 3D}}.
\end{align}
  \begin{lemma}{\rm (rank-one convexity and multiplicative decomposition)}\label{newlemma}
  If  the elastic energy $F\mapsto W(F)$ is rank-one convex,  it follows that the elasto-plastic formulation
\begin{align}
F\mapsto W(F,F_p):={W}(F\,F_p^{-1})={W}(F_e)
\end{align}
 remains rank-one convex w.r.t $F$ \cite{HutterSFB02,Neff00b,LNPzamp2013} at given plastic distortion $F_p$.
 \end{lemma}
 \begin{proof}
 This is clear, because
 \begin{align}
 D_F^2[W(F\,F_p^{-1})].\,(\xi\otimes\eta,\xi\otimes\eta)=D_{F_e}^2[W(F_e)].\,((\xi\otimes {\eta})F_p^{-1},(\xi\otimes
 {\eta})F_p^{-1})=D_{F_e}^2[W(F_e)].\,(\xi\otimes \widehat{\eta},\xi\otimes \widehat{\eta}),\notag
 \end{align}
 where $\widehat{\eta}=F_p^{-T}\eta$.
\end{proof}
 \begin{remark}
 The same constitutive invariance property is true for convexity, polyconvexity and quasiconvexity \cite{HutterSFB02,krishnan2014polyconvex}.
 \end{remark}
Therefore,  the multiplicative approach is ideally suited as far as preservation of ellipticity properties for elastic unloading is concerned. Note that this feature is not true for some additive approaches, see \cite{NeffGhibaAdd}.

\begin{definition}{\rm (reduced dissipation inequality-thermodynamic consistency)}
We say that the reduced dissipation inequality along the plastic evolution is satisfied if and only if
\begin{align}
\frac{\rm d}{\rm dt}[W(F\,F_p^{-1}(t)]\leq 0,
\end{align}
for all constant in time $F$.
\end{definition}

Let us further remark that for fixed $F$ and for an energy for which the decomposition into isochoric and volumetric parts
\begin{align}
W=W_{\rm iso} (F_e)+W_{\rm vol}(F_e)=W_{\rm iso} (F F_p^{-1})+W_{\rm vol}(F)
\end{align}
holds true, in view of  Sansour's result \cite{sansour2008physical} (see also \cite{miehe2014variational1,miehe2014variational,miehe1994representation} and \cite[page
305]{Simo98b}), we have for the reduced dissipation inequality
\begin{align}\label{dissinegW}
 \frac{\rm d}{{\rm d t}}[W_{\rm iso}(F F^{-1}_p)]&=\langle D_{F_e} W_{\rm iso}(F_e),F\frac{\rm d}{\rm dt}[F^{-1}_p]\rangle=\langle D_{F_e} W_{\rm
 iso}(F_e),FF^{-1}_pF_p\frac{\rm d}{\rm dt}[F^{-1}_p]\rangle
\\
&=\langle F_e^T D_{F_e} W_{\rm iso}(F_e),F_p\frac{\rm d}{\rm dt}[F^{-1}_p]\rangle=\langle \Sigma_{e},F_p\frac{\rm d}{\rm dt}[F^{-1}_p]\rangle=-\langle
\Sigma_{e},\frac{\rm d}{\rm dt}[F_p]F_p^{-1}\rangle\leq 0,\notag
\end{align}
where $$\Sigma_{e}=F_e^T D_{F_e} W_{\rm iso}(F_e)=2\,C_e\,D_{ C_e}[\widehat{W}_{\rm iso}(C_e)]=F_e^T \tau_e F_e^{-T}$$ is the mixed variant (transformed)
Kirchhoff tensor
and $$\tau_e:=2\,D_{B_e}[\widetilde{W}_{\rm iso}(B_e)]B_e=2\,D_{\log B_e}[\widetilde{W}_{\rm iso}(\log B_e)]=
D_{\log V_e}[{W}_{\rm iso}(\log V_e)]$$ is the elastic Kirchhoff stress-tensor. Note that $\Sigma_{e}$ is symmetric in case of elastic isotropy, while
$\tau_e$ is always symmetric. The tensor $\Sigma=C\cdot S_2(C)$, where  $S_2=2\,D_C[W(C)]$ is
 the second  Piola-Kirchhoff stress tensor, is sometimes called the \textbf{Mandel stress tensor}  and $\dev_3 \Sigma_{e}=\dev_3 \Sigma_{\rm E}$, where  $\Sigma_{\rm E}$ is the elastic
\textbf{Eshelby tensor}
$$
\Sigma_{\rm E}:=F_e^TD_{F_e}[W({F_e})]-W(F_e)\cdot \id=D_{\log C_e}[\overline{W}(\log C_e)]-\overline{W}(\log C_e)\cdot \id,
$$
 driving the plastic evolution (see e.g. \cite{gupta2007evolution,neff2009notes,maugin1994eshelby,cleja2000eshelby,cleja2003consequences,cleja2013orientational}).

A  simple thermodynamically admissible perfect plasticity model \cite[page 67]{Miehe92}(see also
\cite{Simo98b,Neff01d,Neff01c,Neff_Cosserat_plasticity05,neff2009notes}) is obtained by defining the plastic evolution
\begin{equation}\label{choicchi}
F_p\,\frac{\rm d}{{\rm d t}}[F_p^{-1}]=-\frac{\rm d}{{\rm d t}}[F_p]\,F_p^{-1}\in-\Partial \Chi(\dev_3 \Sigma_{e}),
\end{equation}
where $\Partial \Chi$ is the subdifferential of the indicator function $\Chi$ of the convex elastic domain
 \begin{equation}
\mathcal{E}_{\rm e}({\Sigma_{e}},\frac{1}{3}{\boldsymbol{\sigma}}_{\!\mathbf{y}}^2):=\{\Sigma_{e}\in {\rm Sym}(3)\,|\, \|\dev_3\Sigma_{e}\|^2\leq
\frac{1}{3}{\boldsymbol{\sigma}}_{\!\mathbf{y}}^2\}
\end{equation}
in the mixed-variant $\Sigma_{e}$-stress space\footnote{The choice of $\frac{1}{3}$ in versus $\frac{2}{3}$ is not accidentally, see Lemma \ref{lemmaplast}.}.

The choice \eqref{choicchi} ensures
$
\frac{\rm d}{{\rm d t}}[W_{\rm iso}(F F^{-1}_p)]\leq 0
$
 at fixed $F$, therefore the reduced dissipation inequality \eqref{dissinegW} is satisfied and the deviatoric formulation  together with the use of  $F_p\,\frac{\rm d}{{\rm d t}}[F_p^{-1}]=-\frac{\rm d}{{\rm d t}}[F_p]\, F_p^{-1}$ as conjugate variable guarantees $\det F_p=1$.

\medskip

 Next, a (for us at first surprising) algebraic estimate is introduced.

\begin{lemma}\label{lemmaplast}
Let $F_e\in {\rm GL}^+(3)$ be given. Then it holds $ \|F_e^TSF_e^{-T}\|^2\geq \frac{1}{2}\|S\|^2$ for all  $S\in {\rm Sym}(3)$, the constant being independent of $F_e$.
\end{lemma}
\begin{proof}
Let us define the left Cauchy-Green tensor $B_e=F_eF_e^{T}\in {\rm PSym}(3)$. We have
\begin{equation}
\|F_e^TSF_e^{-T}\|^2=\langle F_e^TSF_e^{-T},F_e^TSF_e^{-T}\rangle=\langle F_e F_e^TS,SF_e^{-T}F_e^{-1}\rangle=\langle B_eS,SB_e^{-1}\rangle.
\end{equation}
Since $B_e=F_eF_e^{T}\in {\rm PSym}(3)$, there is $Q\in {\rm SO}(3)$ such that $D_e=Q^TB_e Q=\diag (d_1,d_2,d_3)$. Hence, we obtain
\begin{equation}
\|F_e^TSF_e^{-T}\|^2=\langle Q^TD_eQS,S Q^T D_e^{-1}Q\rangle=\langle D_eQSQ^T,QS Q^T D_e^{-1}\rangle.
\end{equation}
Moreover,  considering $QSQ^T:=\widehat{S}=\left(
                         \begin{array}{ccc}
                           \widehat{S}_{11} & \widehat{S}_{12}& \widehat{S}_{13} \\
                           \widehat{S}_{12} & \widehat{S}_{22} & \widehat{S}_{23} \\
                           \widehat{S}_{13} & \widehat{S}_{23} & \widehat{S}_{33} \\
                         \end{array}
                       \right)$ and using that $1\leq\frac{d_i}{d_j}+\frac{d_j}{d_i}$ for $d_i,d_j>0$, we deduce
                   \begin{align}
\|F_e^TSF_e^{-T}\|^2&=\widehat{S}_{11}^2+\widehat{S}_{22}^2+\widehat{S}_{33}^2+\widehat{S}_{12}^2\left(\frac{d_1}{d_2}+\frac{d_2}{d_1}\right)
+\widehat{S}_{23}^2\left(\frac{d_2}{d_3}+\frac{d_3}{d_2}\right)
+\widehat{S}_{13}^2\left(\frac{d_1}{d_3}+\frac{d_3}{d_1}\right)\geq\frac{1}{2}\|\widehat{S}\|^2=\frac{1}{2}\|{S}\|^2,\notag
\end{align}
due to the symmetry of $\widehat{S}$ and the proof is complete.
\end{proof}

\begin{remark}
Note that for $S\not\in {\rm Sym}(3)$ there is always an  estimate similar to that given by Lemma
\ref{lemmaplast} but which involves constants depending on $F_e$, i.e. $\|F_e^TSF_e^{-T}\|^2\geq c(F_e)\,\|S\|^2$.
\end{remark}
\begin{remark}
It is easy to see  the relations
\begin{align}\label{7.80} \Sigma_{e}=F_e^T\tau_e \, F_e^{-T}.
\end{align}
Note that \eqref{7.80}, as opposed to appearance,  is not at variance with symmetry of  $\Sigma_{e}$ in case of isotropy.
\end{remark}

 \begin{remark}
Since
\begin{align}
\dev_3 \Sigma_{e}=\dev_3(F_e^T \tau_eF_e^{-T})&=F_e^T \tau_eF_e^{-T}-\frac{1}{3}{\rm tr}(F_e^T \tau_eF_e^{-T})\cdot \id=
F_e^T (\tau_e-\frac{1}{3}{\rm tr} (\tau_e))\cdot \id)F_e^{-T},
\end{align}
we can note
\begin{align}
\dev_3 \Sigma_{e}=F_e^T (\dev_3\tau_e)F_e^{-T}, \qquad \dev_3 \tau_e=F_e^{-T} (\dev_3\Sigma_{e})F_e^{T}, \qquad {\rm tr}(\Sigma_{e})={\rm tr}(\tau_e).
\end{align}
\end{remark}
Thus, using Lemma \ref{lemmaplast} we obtain the estimate
\begin{equation}
 \|\dev_3 \Sigma_{e}\|=\|F_e^T(\dev_3\tau_e)F_e^{-T}\|\geq \frac{1}{\sqrt{2}}\|\dev_3\tau_e\|,
\end{equation}
which   { is valid for general anisotropic materials and it} explains our choice of factors in $\mathcal{E}_{\rm e}({\Sigma_{_{\rm eH}}},\frac{1}{3}{\boldsymbol{\sigma}}_{\!\mathbf{y}}^2)$ and $\mathcal{E}_{\rm e}({\tau_{_{\rm eH}}},\frac{2}{3}{\boldsymbol{\sigma}}_{\!\mathbf{y}}^2)$, respectively. Moreover, numerical tests suggest that the LH-ellipticity domain of the distortional energy function
$
F\mapsto W_{_{\rm eH}}^{\rm iso}(F)=\frac{\mu}{k}\,e^{k\,\|\dev_3\log U\|^2}$, $ F \in{\rm GL}^+(3),
$
with $k\geq \frac{3}{16}$ (the necessary  condition for separate convexity  (SC) of $e^{k\,\|\dev_3\log U\|^2}$ in 3D) is  an extremely  large cone
\begin{align}
\mathcal{E}(W_{_{\rm eH}},{\rm LH},U, 27)=\{U\in{\rm PSym}(3)\,\big|\, \|\dev_3\log U\|^2<27\};
\end{align}
Therefore  we have the inclusion of domains
$$\mathcal{E}_{\rm e}({\Sigma_{_{\rm eH}}},\frac{1}{3}{\boldsymbol{\sigma}}_{\!\mathbf{y}}^2)
    \overset{{\rm Lemma \,\ref{lemmaplast}}}{\subseteq}\mathcal{E}_{\rm e}(\tau_{_{\rm eH}},\frac{2}{3}\, {\boldsymbol{\sigma}}_{\!\mathbf{y}}^2)\ \, \!\!\!\! \!\!\!\!\underbrace{\subseteq}_{\text{numerical tests}} \!\!\!\! \!\!\!\!\ \,  \mathcal{E}(W_{_{\rm eH}},{\rm LH},U, 27).
    $$

\begin{figure}[h!]
\centering
\begin{minipage}[h]{0.7\linewidth}
\includegraphics[scale=0.75]{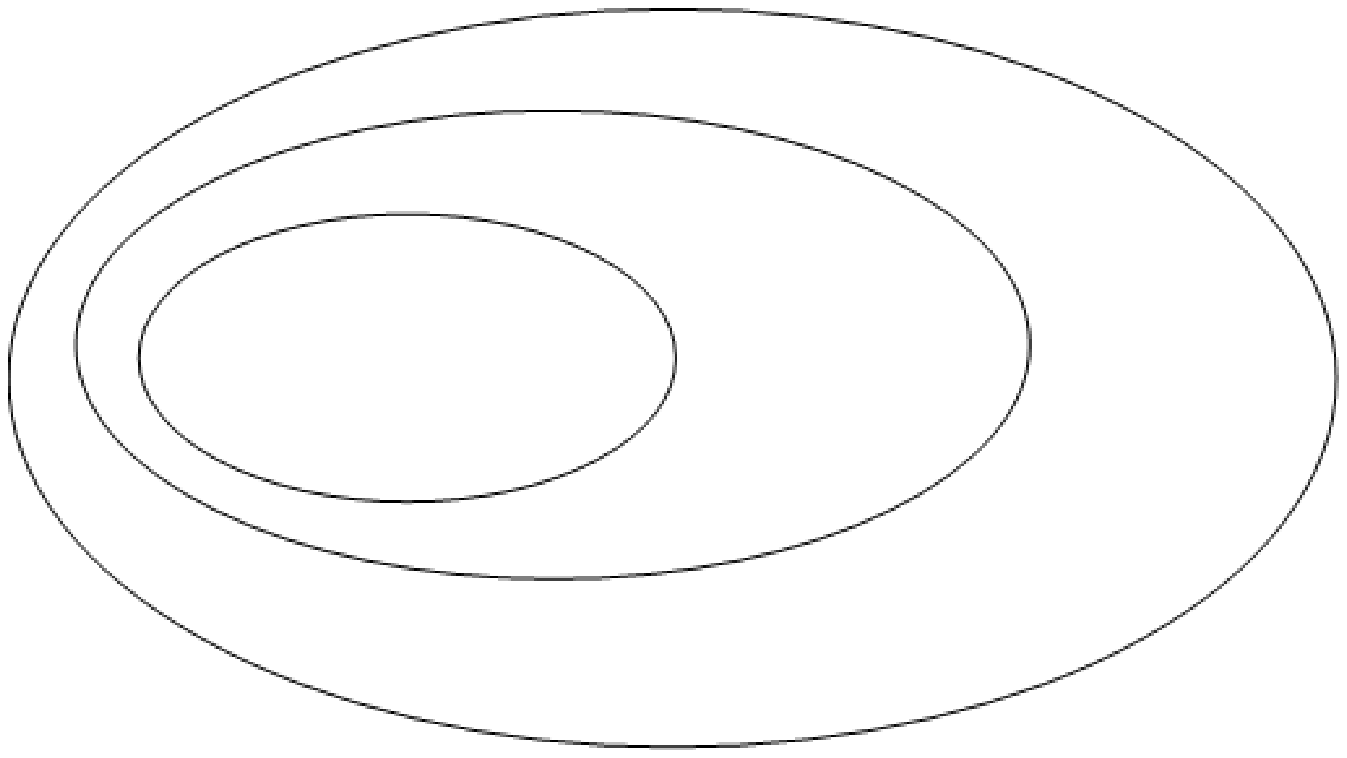}
\put(-95,76){\tiny $\mathcal{E}(W_{_{\rm eH}},{\rm LH},U, 27)$}
\put(-95,68){\tiny\textbf{ellipticity domain} }
\put(-95,60){\tiny of $W_{_{\rm eH}}$}
\put(-163,90){\tiny $\mathcal{E}_{\rm e}(\tau_{_{\rm eH}},\frac{2}{3}\, {\boldsymbol{\sigma}}_{\!\mathbf{y}}^2)$}
\put(-163,82){\tiny \textbf{monotonicity }}
\put(-163,76){\tiny \textbf{domain of} $W_{_{\rm eH}}$}
\put(-250,100){\tiny $\mathcal{E}_{\rm e}({\Sigma_{_{\rm eH}}},\frac{1}{3}{\boldsymbol{\sigma}}_{\!\mathbf{y}}^2)$}
\put(-250,92){\tiny \textbf{elastic domain of} $W_{_{\rm eH}}$}
\centering
\caption{\footnotesize{Elastic domains expressed in the mixed  variant symmetric stress tensor ${\Sigma_{_{\rm eH}}}$  and  the symmetric Kirchhoff stress tensor ${\tau_{_{\rm eH}}}$ related to the ellipticity domain $\mathcal{E}(W_{_{\rm eH}},{\rm LH},U, 27)$. }}
\label{grafic-plasth}
\end{minipage}
\end{figure}%

    In  \eqref{contintau2} the   considered convex ``elastic domain", in which monotonicity and/{or}
ellipticity for $W_{_{\rm eH}}$ is considered, is defined in terms of $\|\dev_3\tau_{_{\rm eH}}\|$, not in terms of $\|\dev_3\Sigma_{_{\rm eH}}\|$.
 However, adapting Lemma
\ref{lemmaplast} to $W_{_{\rm eH}}$, we see that, in the three-dimensional case, our previous results indicating the loss of ellipticity only for extreme
distortional strains suggest that the coupling with plasticity ist most natural: permanent deformation sets in, based on a criterion of distortional energy
(Huber-Hencky-von Mises-type) $\|\dev_3 \Sigma_{_{\rm eH}}\|^2\leq \frac{1}{3}{\boldsymbol{\sigma}}_{\!\mathbf{y}}^2$, and our former results suggest that
$W_{_{\rm eH}}(F_e)$ never reaches the non-elliptic domain in any elasto-plastic process. This is in sharp contrast to the loss of ellipticity of  the quadratic  Hencky energy $W_{_{\rm H}}$, which is
not related to the distortional energy alone. As it turns out, for the overall non-elliptic energy $W_{_{\rm eH}}$ (in three dimensions) plasticity provides a natural relaxation mechanism, which prevents loss of ellipticity in the elastic domain. Moreover, in the above defined elastic domain
$\mathcal{E}_{\rm e}({\Sigma_{_{\rm eH}}},\frac{1}{3}{\boldsymbol{\sigma}}_{\!\mathbf{y}}^2)$, the constitutive relation  $\log B_e\mapsto\sigma(\log B_e)$
remains monotone, i.e. the true-stress-true-strain monotonicity condition (TSTS-M$^+$) is satisfied in $\mathcal{E}_{\rm e}({\Sigma_{_{\rm eH}}},\frac{1}{3}{\boldsymbol{\sigma}}_{\!\mathbf{y}}^2)$ (see \cite{NeffGhibaLankeit,vallee1978lois,vallee2008dual,jog2013conditions}).
  {\begin{remark}\label{remarktauesigmae}
For the isotropic case we have $\tau_e\, B_e=B_e \, \tau_e$, which implies
\begin{align}
\|\dev_3 \Sigma_e\|^2&
=\langle F_e^T\, (\dev_3\tau_e) \, F_e^{-T},F_e^T\, (\dev_3\tau_e) \, F_e^{-T}\rangle=\langle B_e\, (\dev_
3\tau_e) , (\dev_3\tau_e) \, B_e^{-1}\rangle=\|\dev_3\tau_e\|^2.\notag
\end{align}
This fact can be also proved using the fact that, in the isotropic case, both tensors $\Sigma_e$ and $\tau_e$ are symmetric and they have the same invariants. Therefore, using \eqref{7.80} we obtain that  in the isotropic case we have
$\sqrt{2}\, \mathcal{E}_{\rm e}({\Sigma_{_{\rm eH}}},\frac{1}{3}{\boldsymbol{\sigma}}_{\!\mathbf{y}}^2)
   =\mathcal{E}_{\rm e}(\tau_{_{\rm eH}},\frac{2}{3}\, {\boldsymbol{\sigma}}_{\!\mathbf{y}}^2).
    $
    \end{remark}}
\medskip

Summarizing the properties of the 9-dimensional flow  rule for the plastic distortion \eqref{choicchi}  we have:
\begin{itemize}
\item[i)] it is thermodynamically correct ($\frac{\rm d}{\rm d t}[W(F\, F_p^{-1})]\leq 0$);
\item[ii)] the right hand side is  a function of $F$ and $F_p^{-1}$;
\item[iii)]  plastic incompressibility:  the constraint $\det F_p(t)=1$ for all $t\geq 0$ follows from the flow rule, $F_p(t)\in{\rm GL}^+(3)$ for all $t\geq 0$;
\item[iv)] the above properties imply that the flow rule \eqref{choicchi} is consistent;
\item[v)] it satisfies the principle of maximum dissipation and is an associated plasticity model;
\item[vi)] elastic unloading remains rank-one convex under arbitrary plastic predeformation.
\end{itemize}

\section{Conclusion and open problems}

We have shown that the multiplicative plasticity models preserve ellipticity in purely elastic processes at frozen plastic variables provided that the initial elastic response is elliptic.  Preservation of LH-ellipticity is, in our view, a property which should be satisfied by any hyperelastic-plastic model since the elastically unloaded material specimen should respond reasonable under further purely elastic loading. In contrast to multiplicative models, the much used additive logarithmic model does not preserve LH-ellipticity in general \cite{NeffGhibaAdd}.

An interesting question concerns the requirements that one should impose on the elastic response for arbitrary large distortional strains. One may reasonably argue that these requirements are void of any relevance, since the material can never be observed in a state of large distortional strain: prior to that, dissipative processes will occur. In the case of the energy $W_{_{\rm eH}}$, which is not rank-one elliptic for extreme distortional  strains, we have explicitly shown that elastic unloading will remain rank-one convex and the true-stress-true-strain relation remains monotone. This is a remarkable feature in geometrically nonlinear material models. In a future contribution we will provide the analytical proof for the rank-one convexity domain for $W_{_{\rm eH}}$ in $n=3$.

\bibliographystyle{plain} %plain
\addcontentsline{toc}{section}{References}
\begin{footnotesize}

\end{footnotesize}

\end{document}